\documentclass[twocolumn,  groupedaddress, superscriptaddress,  amsmath,amssymb,  aps, pra, ]{revtex4-2}
\usepackage{amsmath,amssymb}
\usepackage{graphicx}
\usepackage{dcolumn}
\usepackage{bm}
\usepackage[usenames, dvipsnames]{color}
\usepackage{float}
\usepackage{mathrsfs}
\usepackage{microtype}
\usepackage{amsthm}
\usepackage[hyperindex=true,pdftitle={Path integral formulation of finite-dimensional quantum mechanics in discrete phase space},
colorlinks=true,pagebackref=false,citecolor=blue,plainpages=false,pdfpagelabels,breaklinks=true,
linkcolor=blue,urlcolor=blue]{hyperref}
\usepackage[margin=2cm]{geometry}

\def\di{{\rm d}}

\def\iu{{\rm i}}

\def\ket#1{|\,#1\, \rangle}

\def\proj#1#2{|\,#1\,\rangle\langle\,#2\,|}

\DeclareMathOperator{\Tr}{Tr}

\newcommand{\Zd}{\mathbb{Z}_d}

\newcommand{\Wfn}{{\rho_\mathtt{W}}}
\newcommand{\Kprop}{{G_\mathtt{W}}}
\newcommand{\Gfn}{{U_\mathtt{W}}}
\newcommand{\Hfn}{H_{\mathtt{W}}}

\newtheorem{theorem}{Theorem}
\newtheorem{proposition}{Proposition}

\begin{document}

\title{Path integral formulation of finite-dimensional quantum mechanics in discrete phase space}

\author{Leonardo A. Pach\'on}
\affiliation{guane Enterprises, R+D+I Unit, Medell\'in 050010, Colombia}

\author{Andr\'es D. G\'omez}
\affiliation{
Grupo de F\'isica At\'omica y Molecular, Instituto de F\'{\i}sica, Facultad de Ciencias Exactas y Naturales, 
Universidad de Antioquia UdeA; Calle 70 No.~52-21, Medell\'in, Colombia.}

\date{\today}

\begin{abstract}
We develop a path integral representation for the dynamics of quantum
systems with a finite-dimensional Hilbert space, formulated entirely
within a discrete phase space.  Starting from the discrete Wigner
function on $\Zd \times \Zd$ ($d$ an odd prime) and the associated Weyl
transform built from generalized displacement operators, we derive an
exact kernel that propagates the discrete Wigner function in time and,
by iterating its composition law through a short-time approximation,
obtain a sum-over-paths expression weighted by a discrete phase-space
action---the natural finite-dimensional counterpart of Marinov's
functional.  For Hamiltonians linear in the phase-space coordinates and
at times strictly commensurate with the lattice, the fluctuation sum
collapses to a deterministic shift, realizing the discrete analog of
classical Hamiltonian flow.  Applying the formalism to one and to two
interacting qutrits ($d=3$), we show that the full entanglement
dynamics---captured by a closed-form linear entropy valid for all
times---requires the coherent contribution of all fluctuation sectors;
the boundary-term (mean-field) sector alone fails to reproduce it.  For
a non-stabilizer Hamiltonian, where the short-time kernel is only
approximate, the time-sliced path integral converges to the exact
dynamics, including the dynamical generation of Wigner negativity.  We
discuss implications for the semiclassical simulation of many-body spin
systems and for the characterization of non-classicality through Wigner
negativity.
\end{abstract}

\maketitle

\section{Introduction}
\label{sec:intro}
 
Many physically important quantum systems are naturally described by
finite-dimensional Hilbert spaces. Nuclear and electronic spin degrees of
freedom, atomic-level schemes in quantum optics, and---most prominently---qubits
and qudits in quantum information processing all belong to this category.
Despite the central role these systems play in modern physics, their
phase-space formulation has received comparatively less attention than the
well-established Wigner--Weyl--Moyal framework for continuous-variable
systems~\cite{Wigner1932,Weyl1927,Moyal1949,Groenewold1946,Hillery1984}.
 
For continuous degrees of freedom, the Wigner function provides a
complete quasi-probability representation of quantum states on classical
phase space. The dynamical evolution of the Wigner function can be
expressed through an evolution kernel, and Marinov~\cite{Marinov1991}
showed that this kernel admits a path integral representation over phase
space in which a new functional---distinct from but related to the
Hamiltonian action---appears with real weights.  This path integral has
proven valuable both conceptually, in connecting quantum and classical
dynamics, and computationally, as the basis for semiclassical
approximations~\cite{Dittrich2006,Dittrich2010,OzorioDeAlmeida2006,DittrichPachon2009}.
 
The extension of the Wigner function to finite-dimensional systems was
pioneered by Buot~\cite{Buot1974}, Hannay and Berry~\cite{HannayBerry1980},
Cohen and Scully~\cite{CohenScully1986}, and Feynman~\cite{Feynman1987}
for specific low-dimensional cases, and formulated systematically by
Wootters~\cite{Wootters1987} for systems with a prime number $d$ of
orthogonal states. In Wootters's construction, the phase space is a
discrete $d\times d$ lattice with toroidal topology, and the Wigner
function is a real-valued distribution whose marginals reproduce the
correct quantum probabilities. This framework was subsequently extended
to prime-power dimensions by Gibbons, Hoffman, and
Wootters~\cite{Gibbons2004}, and explored by many
authors~\cite{Leonhardt1995,GalettiPiza1988,Vourdas2004,BjorkKlimovSanchez2008}.
 
A central result that underscores the physical significance of the
discrete Wigner function is Gross's discrete Hudson
theorem~\cite{Gross2006}: for odd-prime-dimensional systems, the only
pure states with a non-negative discrete Wigner function are the
stabilizer states.  Since stabilizer circuits can be efficiently
simulated classically~\cite{Gottesman1998,AaronsonGottesman2004}, 
Wigner negativity emerges as a necessary resource for quantum 
computational advantage within the quasiprobability-based simulation 
framework. This connection has been made precise through the resource 
theory of magic states~\cite{Veitch2012,Howard2014,Mari2012,Pashayan2015}
and reviewed within the broader quasi-probability landscape in~\cite{Ferrie2011};
it also has deep ties to
contextuality~\cite{Raussendorf2017,Delfosse2015}.
 
On the computational side, the discrete Wigner function has enabled powerful
semiclassical simulation methods for many-body spin systems.  
Schachenmayer, Pikovski, and Rey~\cite{Schachenmayer2015,Schachenmayer2015b}
introduced the discrete truncated Wigner approximation (DTWA), in which
classical trajectories on the discrete phase space are sampled via Monte
Carlo to approximate the quantum dynamics of interacting spin lattice
models.  This method has been successfully applied to long-range
interacting systems in regimes inaccessible to exact diagonalization or
tensor-network methods~\cite{Acevedo2017,Zhu2019}.
 
While the discrete Wigner function and its static properties are now
well understood, a systematic treatment of the \emph{dynamics}---in
particular a Marinov-style path integral formulation of the evolution 
kernel in discrete phase space with an explicit action 
functional---has not been fully developed.  A closely related line of 
work by Kocia and Love~\cite{KociaLove2021} develops a periodized 
stationary-phase method for discrete Wigner functions of odd-prime 
dimension, using $p$-adic number-theoretic techniques to evaluate 
quadratic Gauss sums; their approach is targeted at classical 
simulation of quantum circuits and complements the present 
Marinov-style construction, which is oriented toward Hamiltonian 
dynamics and many-body semiclassical methods.  Such a formulation 
would provide a natural bridge between the continuous Marinov path
integral and the discrete phase-space methods used in quantum
information and many-body physics, and points toward a route for
incorporating the quantum corrections that mean-field schemes such as
the DTWA omit---though, as we discuss in
Sec.~\ref{sec:applications}, turning this into a tractable many-body
method requires taming the sign problem of the oscillatory path sum.
 
In this paper we address this gap. We construct the discrete Weyl
transform and Wigner function following the algebraic approach based on
generalized displacement operators, derive the evolution kernel of the
discrete Wigner function, and obtain a path integral representation for
this kernel by iterating the composition law. Our main result,
Eq.~(\ref{eq:pathintegral}), expresses the propagator as a sum over
paths on $\Zd \times \Zd$ weighted by a phase-space action that is
the discrete counterpart of Marinov's functional.  Throughout the paper
we restrict to $d$ an odd prime, where $\Zd$ is a finite field admitting
the multiplicative inverse of $2$.  We illustrate the formalism with a
single-qutrit ($d=3$) example and apply it to an interacting two-qutrit
system, where we show that the nonzero fluctuation sectors of the path
integral are precisely the ingredients missing from any truncation
restricted to the boundary-term (mean-field) sector.

It is worth delineating what is and is not new.  The discrete twisted
convolution that underlies the kernel (Sec.~\ref{sec:kernel}) is the
discrete Moyal star product, whose explicit forms have been studied
before~\cite{GalettiPiza1988,Vourdas2004,BjorkKlimovSanchez2008}, and
the discrete Wigner propagator can be written compactly in terms of
phase-point operators.  The contribution of the present work is, rather,
(i) the assembly of these ingredients, through the kernel composition
law, into a genuine sum-over-paths with an \emph{explicit} discrete
action functional~[Eq.~(\ref{eq:discreteaction})] in one-to-one
correspondence with Marinov's continuous construction; (ii) the
identification of a strictly commensurate ``pseudo-classical'' regime in
which the fluctuation sum collapses to a Clifford shift; (iii) the
demonstration---both at finite time step and in the continuum
limit---that the boundary (mean-field) sector alone cannot reproduce the
dynamics, so that the $\tilde\mu\neq0$ sectors are indispensable; and
(iv) the explicit verification that the time-sliced sum converges to the
exact non-stabilizer dynamics.  Unlike the circuit-oriented
stationary-phase method of Kocia and Love~\cite{KociaLove2021}, which
evaluates discrete Wigner propagators of Clifford-plus-$T$ circuits, the
object constructed here is a continuous-time action for Hamiltonian
evolution.

The paper is organized as follows.  Section~\ref{sec:continuous}
reviews the Marinov path integral for continuous systems.
Section~\ref{sec:discrete} develops the discrete Weyl transform
for systems with a $d$-dimensional Hilbert space ($d$ odd prime).
Section~\ref{sec:kernel} derives the evolution kernel of the discrete
Wigner function.  Section~\ref{sec:pathintegral} obtains the path
integral representation.  Section~\ref{sec:qutrit} illustrates the
formalism with qutrit examples, including an interacting two-qutrit
system. Section~\ref{sec:applications} discusses applications. 
Section~\ref{sec:conclusions} presents conclusions and outlook.

\section{Continuous Phase-Space Formulation}
\label{sec:continuous}
 
We briefly recall the Wigner--Weyl formalism for continuous degrees of
freedom and the Marinov path integral, which serve as the template for
the discrete construction.
 
For a system with one continuous degree of freedom, the Weyl
transform~\cite{Weyl1927,Groenewold1946} maps a phase-space function
$f(q,p)$ to a self-adjoint operator $\hat{A}_f$ via the displacement
operators $\mathcal{U}(u) = e^{-\iu u\hat{P}/\hbar}$ and
$\mathcal{V}(v) = e^{-\iu v\hat{Q}/\hbar}$, which satisfy
the Weyl relation
$\mathcal{U}(u)\mathcal{V}(v) =
e^{\iu uv/\hbar}\mathcal{V}(v)\mathcal{U}(u)$.
A symmetric ordering factor ensures self-adjointness, and the inverse
Weyl transform (the Wigner map) associates to any density matrix
$\hat{\rho}$ a real quasi-probability distribution---the Wigner
function---on phase space~\cite{Wigner1932,Hillery1984,Case2008,CurtrightZachos2012}.
 
Under unitary evolution
$\hat{\rho}(t) = \hat{U}(t)\hat{\rho}(0)\hat{U}^\dagger(t)$
with $\hat{U}(t) = e^{-\iu\hat{H}t/\hbar}$, the Wigner function
evolves as
\begin{equation}
\label{eq:Wevolution}
\rho_\mathrm{W}(\mathbf{r}'',t) = \int \di \mathbf{r}'\,
G_\mathrm{W}(\mathbf{r}'',t; \mathbf{r}',0)\,\rho_\mathrm{W}(\mathbf{r}',0)\,,
\end{equation}
where $\mathbf{r} = (q,p)$ and $\mathcal{G}_\mathrm{W}$ is the Wigner
propagator~\cite{Marinov1991,Dittrich2006}. Marinov~\cite{Marinov1991}
showed that $G_\mathrm{W}$ admits a path integral representation
\begin{equation}
\label{eq:Marinov}
G_\mathrm{W}(\mathbf{r}'',\mathbf{r}') =
\frac{1}{(2\pi\hbar)^f}\int\!\mathcal{D}^{2f}\!r
\int\!\mathcal{D}^{2f}\!\tilde{r}\;
e^{\frac{\iu}{\hbar}S_{\mathrm{W}}[\mathbf{r},\tilde{\mathbf{r}},t]}\,,
\end{equation}
with the phase-space action functional
\begin{equation}
\label{eq:Marinovaction}
S_{\mathrm{W}}[\mathbf{r},\tilde{\mathbf{r}},t] = \int_0^t \!\left[
\dot{\mathbf{r}}\wedge\tilde{\mathbf{r}}
+ H_{\mathrm{W}}\!\left(\mathbf{r}+\tfrac{1}{2}\tilde{\mathbf{r}}\right)
- H_{\mathrm{W}}\!\left(\mathbf{r}-\tfrac{1}{2}\tilde{\mathbf{r}}\right)
\right]\di t'\,.
\end{equation}
Here $\mathbf{r}_i\wedge\mathbf{r}_j =
\mathbf{r}_i^{\mathrm{T}}\mathsf{J}\,\mathbf{r}_j$ is the symplectic
product, $\mathbf{r}(t')$ a trajectory with fixed endpoints, and
$\tilde{\mathbf{r}}(t')$ an unrestricted fluctuation variable. The
extrema of $S$ coincide with the classical Hamilton equations.  This
path integral has been used for semiclassical
propagation~\cite{Dittrich2006,Dittrich2010,OzorioDeAlmeida2006} and
forms the basis for Monte Carlo methods in phase
space~\cite{Marinov1991}.
 
The goal of this paper is to construct the discrete counterpart of
Eqs.~(\ref{eq:Marinov})--(\ref{eq:Marinovaction}) for systems with a
finite-dimensional Hilbert space.
 
\section{Discrete Phase Space and Weyl Transform}
\label{sec:discrete}
 
For quantum systems with a finite-dimensional Hilbert space
$\mathscr{H}$ of dimension $d$, a classical phase space in the usual
sense does not exist.  Nevertheless, the algebraic structure reviewed in
Sec.~\ref{sec:continuous} provides a blueprint: the Weyl transform is
built from displacement operators, which in turn are constructed from
position and momentum operators related by a Fourier transform.  As we
now show, each of these ingredients has a natural discrete analog.
 
Throughout this section, and unless stated otherwise, we take $d$ to be
an odd prime.  We denote by $\omega = e^{2\pi \iu / d}$ a primitive
$d$th root of unity, and all arithmetic in the subscripts of $\omega$
is understood modulo~$d$.
 
\subsection{The discrete Fourier transform and conjugate operators}
\label{sec:fourier}
 
The starting point is the observation that the Fourier transform plays a
dual role in continuous quantum mechanics: it is both a mathematical tool
and the unitary transformation relating the position and momentum
representations.  The Stone--von Neumann
theorem~\cite{VonNeumann1931,Stone1930} guarantees that, up to unitary
equivalence, there is a unique irreducible representation of the
canonical commutation relations $[\hat{Q},\hat{P}] = \iu\hbar$, and the
Fourier transform implements the change of basis between the two.
 
For finite-dimensional systems, we proceed by analogy.  Let
$\{\ket{x,n}\}_{n\in\Zd}$ be an orthonormal basis of ``position-like''
states.  For a qutrit system ($d = 3$), one may take
$\ket{x,n}$ to be the eigenstates of a diagonal operator
$\hat{x}$ with eigenvalues $0, 1, 2$.
 
We define the \emph{discrete Fourier transform} (DFT) as the unitary
operator $\mathscr{F}: \mathscr{H} \to \mathscr{H}$ given by
\begin{equation}
\label{eq:DFT}
\mathscr{F} = \frac{1}{\sqrt{d}} \sum_{m,n=0}^{d-1}
\omega^{mn}\,\proj{x,m}{x,n}\,.
\end{equation}
One readily verifies that $\mathscr{F}$ is unitary,
$\mathscr{F}^\dagger \mathscr{F} = I$, and that $\mathscr{F}^4 = I$.
 
We then define the position-like and momentum-like operators
\begin{equation}
\label{eq:xpops}
\hat{x} = \sum_{n=0}^{d-1} n\,\proj{x,n}{x,n}\,,\qquad
\hat{p} = \sum_{n=0}^{d-1} n\,\proj{p,n}{p,n}\,,
\end{equation}
where the momentum-like basis is
$\ket{p,n} = \mathscr{F}\ket{x,n}$.  By construction,
$\hat{p} = \mathscr{F}\,\hat{x}\,\mathscr{F}^\dagger$, in direct
analogy with the continuous case.
 
\subsection{Generalized displacement operators}
\label{sec:displacement}
 
From the conjugate operators $\hat{x}$ and $\hat{p}$, we construct the
clock and shift operators~\cite{Schwinger1960}
\begin{equation}
\label{eq:ZX}
\hat{Z} = \exp\!\left(\frac{2\pi\iu}{d}\,\hat{x}\right),\qquad
\hat{X} = \exp\!\left(-\frac{2\pi\iu}{d}\,\hat{p}\right).
\end{equation}
These operators act on the position basis as
\begin{equation}
\hat{Z}\ket{x,n} = \omega^n\ket{x,n}\,,\qquad
\hat{X}\ket{x,n} = \ket{x,n\oplus 1}\,,
\end{equation}
where $n \oplus 1$ denotes addition modulo $d$.  They satisfy the Weyl
commutation relation
\begin{equation}
\label{eq:ZXcommutation}
\hat{Z}\hat{X} = \omega\,\hat{X}\hat{Z}\,,
\end{equation}
which is the discrete analog of
$\mathcal{U}(u)\mathcal{V}(v) = e^{\iu
uv/\hbar}\mathcal{V}(v)\mathcal{U}(u)$.
 
The \emph{displacement operators} are defined as
\begin{equation}
\label{eq:displacement}
\hat{D}(k,j) = \omega^{-kj/2}\,\hat{Z}^k\,\hat{X}^j\,,
\qquad (k,j) \in \Zd \times \Zd\,,
\end{equation}
where $kj/2$ is understood as $kj \cdot 2^{-1} \pmod{d}$
(well-defined since $d$ is an odd prime and $\gcd(2,d) = 1$).  The
phase factor $\omega^{-kj/2}$ ensures that
$\hat{D}(k,j)^\dagger = \hat{D}(-k,-j)$.  The $d^2$ operators
$\{\hat{D}(k,j)\}$ form an orthogonal basis for the space of linear
operators on $\mathscr{H}$ with respect to the Hilbert--Schmidt inner
product:
\begin{equation}
\label{eq:orthogonality}
\Tr\!\left[\hat{D}(k,j)^\dagger\,\hat{D}(k',j')\right]
= d\,\delta_{k,k'}\,\delta_{j,j'}\,.
\end{equation}
Using Eq.~(\ref{eq:ZXcommutation}) to move $\hat{X}^{j_1}$ past
$\hat{Z}^{k_2}$, one also verifies the multiplication rule
\begin{equation}
\label{eq:Dmult}
\hat{D}(k_1,j_1)\,\hat{D}(k_2,j_2) =
\omega^{2^{-1}(k_1 j_2 - k_2 j_1)}\,
\hat{D}(k_1+k_2,j_1+j_2),
\end{equation}
consistent with $\hat{D}(k,j)^\dagger = \hat{D}(-k,-j)$.
 
\subsection{The discrete Weyl transform}
\label{sec:discreteweyl}
 
With the operator basis $\{\hat{D}(k,j)\}$ at hand, we define the
\emph{discrete Weyl transform} of an operator $\hat{A}$ acting on
$\mathscr{H}$ as
\begin{equation}
\label{eq:discreteWeyl}
\tilde{A}_{\mathtt{W}}(k,j) = \Tr\!\left[\hat{A}\,\hat{D}(k,j)\right]
= \Tr\!\left[\hat{A}\,\hat{Z}^k\,\hat{X}^j\,\omega^{-kj/2}\right],
\end{equation}
where $(k,j) \in \Zd \times \Zd$.  The function $\tilde{A}_{\mathtt{W}}(k,j)$
lives on the ``reciprocal'' discrete phase space.  The inversion formula
is
\begin{equation}
\label{eq:WeylInversion}
\hat{A} = \frac{1}{d}\sum_{k,j\in\Zd}
\tilde{A}_{\mathtt{W}}(k,j)\,\hat{D}(k,j)^\dagger\,,
\end{equation}
which follows directly from the orthogonality
relation~(\ref{eq:orthogonality}).
 
The corresponding function on the \emph{direct} discrete phase space
$\Zd \times \Zd$ is obtained by applying the two-dimensional
inverse DFT:
\begin{equation}
\label{eq:phasespacefunction}
A_{\mathtt{W}}(m,n) = \mathscr{F}^{-1}[\tilde{A}_{\mathtt{W}}](m,n)
= \frac{1}{d}\sum_{k,j\in\Zd} \tilde{A}_{\mathtt{W}}(k,j)\,
\omega^{jn - km}\,.
\end{equation}
We call $\Zd \times \Zd$, equipped with this correspondence, the
\emph{discrete phase space} of the system.  Its topology is that of a
toroidal lattice, as first noted by Wootters~\cite{Wootters1987}.
 
\subsection{The discrete Wigner function}
\label{sec:discretewigner}
 
Given a density matrix $\hat{\rho}$, we define the \emph{discrete Wigner
function} as the real-valued function
$\Wfn: \Zd \times \Zd \to \mathbb{R}$:
\begin{equation}
\label{eq:DWF}
\Wfn(m,n) = \frac{1}{d^2}\sum_{k,j\in\Zd}
\tilde{\rho}(k,j)\,\omega^{jn - km}\,,
\end{equation}
where $\tilde{\rho}_{\mathtt{W}}(k,j) = \Tr[\hat{\rho}\,\hat{D}(k,j)]$ is the
Weyl symbol of $\hat{\rho}$.  The extra factor of $1/d$ (relative
to Eq.~(\ref{eq:phasespacefunction})) is conventional and ensures that
\begin{equation}
\sum_{m,n\in\Zd} \Wfn(m,n) = 1\,,
\end{equation}
as in Wootters's original formulation~\cite{Wootters1987}.  The
marginal properties are the discrete analogs of the
continuous case~\cite{Wootters1987,Gibbons2004}.
 
Since the propagator $\Kprop$ derived below is constructed to map
$\Wfn\to\Wfn$ with the normalization~(\ref{eq:DWF}), the factor
$1/d^2$ enters both sides of the evolution equation~(\ref{eq:Wevolve})
consistently and plays no role in the qualitative structure of the
kernel (a different choice would simply rescale $\Kprop$ by an overall
constant).
 
It is important to emphasize the difference between this formulation and
those based on systems with continuous degrees of freedom but discrete
symmetry (such as a particle on a ring), where the Hilbert space is
infinite-dimensional and the discreteness arises from approximating the
Fourier transform.  Here, the Hilbert space is intrinsically
$d$-dimensional, and the discrete nature of the Wigner function is a
natural consequence of the algebraic structure.
 
\section{Evolution Kernel of the Discrete Wigner Function}
\label{sec:kernel}
 
We now derive the kernel that propagates the discrete Wigner function
under unitary evolution.
 
\subsection{Time evolution in the Weyl representation}
 
Under unitary evolution with $\hat{U}(t) = e^{-\iu\hat{H}t/\hbar}$,
the density matrix evolves as
$\hat{\rho}(t) = \hat{U}(t)\hat{\rho}(0)\hat{U}^\dagger(t)$.
On the discrete phase space, the time-evolved Wigner function is related
to the initial one by
\begin{equation}
\label{eq:Wevolve}
\Wfn(\mu, t) = \sum_{\mu'\in\Zd\times\Zd}
\Kprop(\mu,t;\mu',0)\,\Wfn(\mu',0)\,,
\end{equation}
where $\mu = (m,n)$ denotes a point on the discrete phase space and
$\Kprop$ is the evolution kernel we seek.
 
To compute $\Kprop$, we first need the composition rule for the Weyl
transform of a product of operators.
 
\begin{proposition}[Discrete twisted convolution]
\label{prop:twistedconv}
Let $d$ be an odd prime, and let $\hat{A}_g$ and $\hat{A}_h$ be operators on $\mathscr{H}$ with
discrete Weyl symbols $\tilde{g}_{\mathtt{W}}(k,j)$ and $\tilde{h}_{\mathtt{W}}(k,j)$,
respectively.  Then the Weyl symbol of
$\hat{A}_f = \hat{A}_g\hat{A}_h$ is
\begin{equation}
\label{eq:twistedconv}
\begin{split}
\tilde{f}_{\mathtt{W}}(k,j) &= \frac{1}{d}\sum_{k',j'\in\Zd}
\omega^{2^{-1}(k'j - kj')} 
\\ & \times
\tilde{g}_{\mathtt{W}}(k+k',j+j')\,\tilde{h}_{\mathtt{W}}(-k',-j')\,,
\end{split}
\end{equation}
where $2^{-1}$ denotes the multiplicative inverse of $2$ in $\Zd$
(i.e., $2^{-1} = (d+1)/2$), which exists since $d$ is odd.
\end{proposition}
 
\noindent\textit{Proof.}
By Eq.~(\ref{eq:discreteWeyl}),
$\tilde{f}_{\mathtt{W}}(k,j) = \Tr[\hat{A}_g\hat{A}_h\,\hat{D}(k,j)]$.  Expanding
$\hat{A}_g$ via Eq.~(\ref{eq:WeylInversion}) gives
\begin{equation}
\tilde{f}_{\mathtt{W}}(k,j) = \frac{1}{d}\sum_{k'',j''\in\Zd}\tilde{g}_{\mathtt{W}}(k'',j'')
\Tr\!\left[\hat{D}(k'',j'')^\dagger\hat{A}_h\hat{D}(k,j)\right].
\end{equation}
By cyclicity of the trace and $\hat{D}(k'',j'')^\dagger=\hat{D}(-k'',-j'')$,
this becomes
$\Tr[\hat{A}_h\hat{D}(k,j)\hat{D}(-k'',-j'')]$.
Applying the displacement-operator multiplication
rule~(\ref{eq:Dmult}) with $(k_1,j_1) = (k,j)$ and $(k_2,j_2) = (-k'',-j'')$,
\begin{equation*}
\hat{D}(k,j)\hat{D}(-k'',-j'') =
\omega^{2^{-1}(k''j-kj'')}\hat{D}(k-k'',j-j''),
\end{equation*}
yields
\begin{equation}
\tilde{f}_{\mathtt{W}}(k,j) = \frac{1}{d}\sum_{k'',j''}
\omega^{2^{-1}(k''j-kj'')}\,\tilde{g}_{\mathtt{W}}(k'',j'')\,\tilde{h}_{\mathtt{W}}(k-k'',j-j'').
\end{equation}
The change of dummy variable $k''\to k+k'$, $j''\to j+j'$ (equivalent
to $k'=-(k-k'')$, $j'=-(j-j'')$) reproduces
Eq.~(\ref{eq:twistedconv}).\hfill$\square$
 
\medskip
 
This is the discrete analog of the Moyal star
product~\cite{Moyal1949}; explicit forms of the star product on discrete 
phase space have been discussed in~\cite{GalettiPiza1988,Vourdas2004,
BjorkKlimovSanchez2008}.
 
\subsection{Derivation of the kernel}
\label{sec:kernelderivation}
 
Since $\hat{\rho}(t) = \hat{U}(t)\hat{\rho}(0)\hat{U}^\dagger(t)$
is a double product, we apply Proposition~\ref{prop:twistedconv}
twice.  Denoting the Weyl symbol of $\hat{U}(t)$ by
$\tilde{\Gfn}(k,j)$ and using the unitarity relation
$\widetilde{U^\dagger}(-k,-j) = \overline{\tilde{\Gfn}(k,j)}$, one
obtains after a straightforward calculation the Weyl-space kernel
$\tilde{\Kprop}(k',j';k,j)$ such that
$\tilde{\rho_t}(k',j') = \sum_{k,j}\tilde{\Kprop}(k',j';k,j)\tilde{\rho_0}(k,j)$,
namely
\begin{multline}
\label{eq:Khat}
\tilde{\Kprop}(k',j';k,j) = \frac{1}{d^2}\!\!\sum_{k'',j''\in\Zd}
\omega^{2^{-1}[k''(j' + j) - j''(k' + k)]}\times\\
\times\tilde{\Gfn}(k'+k'',j'+j'')\,
\overline{\tilde{\Gfn}(k''+k,j''+j)}\,.
\end{multline}
 
Transforming to the direct phase space via the inverse DFT yields
the propagator in closed form.
 
\begin{theorem}[Discrete Wigner propagator]
\label{thm:propagator}
For $d$ an odd prime, the evolution kernel of the discrete Wigner
function is
\begin{multline}
\label{eq:Kdirect}
\Kprop(m',n';m,n) = \frac{1}{d^2}
\sum_{\tilde{m},\tilde{n}\in\Zd}
\omega^{-2[\tilde{m}(n'-n) - \tilde{n}(m'-m)]}\times\\
\times\Gfn(\tilde{m}+m,\tilde{n}+n)\,
\overline{\Gfn(m'-\tilde{m},n'-\tilde{n})}\,,
\end{multline}
where $\Gfn(m,n) = \mathscr{F}^{-1}[\tilde{\Gfn}](m,n)$
is the phase-space representation of the evolution operator.
Equivalently, in symplectic form,
$\Kprop(\mu',\mu) = (1/d^2)\sum_{\tilde\mu}
\omega^{2\Delta\mu\wedge\tilde\mu}\,\Gfn(\mu+\tilde\mu)\,
\overline{\Gfn(\mu'-\tilde\mu)}$, with
$\Delta\mu = \mu'-\mu$ and
$\mu\wedge\tilde\mu = m\tilde n - n\tilde m$.
\end{theorem}
 
\noindent\textit{Proof.}
Starting from~(\ref{eq:Khat}), we Fourier-invert the Weyl-space 
kernel with respect to both $(k',j')$ and $(k,j)$ using
Eq.~(\ref{eq:phasespacefunction}), and substitute the Fourier
representations of $\tilde{\Gfn}$ and $\overline{\tilde{\Gfn}}$.
The sums over $(k',j')$, $(k,j)$, and the internal $(k'',j'')$ are
evaluated using $\frac{1}{d}\sum_{k\in\Zd} \omega^{k\ell} = \delta_{\ell,0}$;
the resulting Kronecker deltas fix two pairs of summation variables
in terms of a single remaining pair $(\tilde{m},\tilde{n})$.
Equation~(\ref{eq:Kdirect}) then follows after the change of
variables $\tilde{m} \to \tilde{m} - m$, $\tilde{n} \to \tilde{n} - n$,
permissible since the sum ranges over all of $\Zd$.\hfill$\square$
 
\medskip
 
Equation~(\ref{eq:Kdirect}) is the discrete analog of the
continuous Wigner propagator.  The phase factor
$\omega^{2\Delta\mu\wedge\tilde\mu}$ plays the role of
$e^{\iu\tilde{\mathbf{r}}\wedge(\mathbf{r}''-\mathbf{r}')/\hbar}$, with
the continuous symplectic product replaced by the discrete symplectic form
on $\Zd \times \Zd$.  The kernel is real-valued.  This can be seen 
operationally by writing 
$\Kprop(\mu',\mu) = d^{-1}\Tr[\hat A(\mu')\,\hat U\,\hat A(\mu)\,\hat U^\dagger]$,
where $\hat A(\mu) = d^{-1}\sum_{k,j}\omega^{km-jn}\hat D(k,j)^\dagger$ are 
the Hermitian phase-point operators~\cite{Wootters1987,Gibbons2004} 
satisfying $\Tr[\hat A(\mu')\hat A(\mu)] = d\,\delta_{\mu',\mu}$ and 
$\hat\rho = \sum_\mu\Wfn(\mu)\hat A(\mu)$; reality then follows as 
the trace of a product of Hermitian operators.  At the level of 
Eq.~(\ref{eq:Kdirect}), the same property is encoded in the 
$\omega$-phase cancellations of the $\tilde\mu$ sum, which are 
non-trivial term-by-term.  We have verified Eq.~(\ref{eq:Kdirect}) 
against direct evolution $\hat\rho(t) = \hat{U}(t)\hat\rho(0)
\hat{U}^\dagger(t)$ to machine precision for $d=3,5,7,11$; all 
numerical claims throughout this paper are reproduced by the code 
repository accompanying this work~\cite{quditMarinov2026}.  For the 
two-qutrit composite system of Sec.~\ref{sec:beyondDTWA} the same 
structure applies with the tensor-product phase-point operators 
$\hat A_\otimes(\mu) = \hat A(\mu_1)\otimes\hat A(\mu_2)$, and the 
kernel Eq.~(\ref{eq:Kdirect}) generalizes by replacing the $1/d^2$ 
prefactor and single $\tilde\mu$-sum by $1/d^{2n}$ and a sum over 
$(\Zd)^{2n}$, respectively (see Eq.~(\ref{eq:pathsum2})).
 
\section{Path Integral Formulation}
\label{sec:pathintegral}
 
The evolution kernel satisfies a composition law that allows us to
construct a sum-over-paths representation.
 
\subsection{Composition law}
 
From Eq.~(\ref{eq:Wevolve}) and the group property of the evolution
operator, it follows immediately that for $t > t' > t_0$,
\begin{equation}
\label{eq:composition}
\Kprop(\mu,t;\mu'',t_0)
= \sum_{\mu'\in\Zd\times\Zd}
\Kprop(\mu,t;\mu',t')\,\Kprop(\mu',t';\mu'',t_0)\,.
\end{equation}
This composition law is the discrete analog of the semigroup property
of the continuous propagator, and is the key ingredient for the path
integral construction.
 
\subsection{Time slicing and short-time kernel}
 
Let us partition the time interval $[0,t]$ into $N$ equal steps of
duration $\tau = t/N$, with $t_i = i\tau$.  Iterating the composition
law~(\ref{eq:composition}) gives
\begin{equation}
\label{eq:slicing}
\Kprop(\mu,t;\mu_0,0) = \sum_{\mu_1,\ldots,\mu_{N-1}}
\prod_{i=1}^{N} \Kprop(\mu_i,t_i;\mu_{i-1},t_{i-1})\,,
\end{equation}
where $\mu_N \equiv \mu$.
 
For small $\tau$, the phase-space function of the short-time evolution
operator $\hat{U}(\tau) = e^{-\iu\hat{H}\tau/\hbar}$ can be
approximated as
\begin{equation}
\label{eq:shorttime}
\Gfn(\mu) \approx \exp\!\left(-\frac{\iu}{\hbar}\,\Hfn(\mu)\,\tau\right),
\end{equation}
where $\Hfn(\mu) = \mathscr{F}^{-1}[\tilde{H}](\mu)$ is the
phase-space representation of the Hamiltonian.
Equation~(\ref{eq:shorttime}) approximates the phase-space symbol of
$e^{-\iu\hat H\tau/\hbar}$ by the exponential of the symbol of
$\hat H$; it is a symbol-level short-time expansion (not the Trotter
splitting $e^{\hat{A} + \hat{B} }\!\!\approx\!\!e^{\hat{A}} e^{\hat{B}}$ at the operator level).
Expanding the exact symbol $\Gfn$ via the discrete Moyal
star-exponential $e^{-\iu\tau\hat H/\hbar}_\star$ one finds
\begin{equation}
\label{eq:shorttimeerror}
\begin{split}
\Gfn(\mu) &= e^{-\iu\tau\Hfn(\mu)/\hbar}
\\ &- \frac{\tau^2}{2\hbar^2}\big(\Hfn\star\Hfn - \Hfn^2\big)(\mu) + O(\tau^3),
\end{split}
\end{equation}
so the quadratic correction is governed by the discrepancy between 
the ordinary square $\Hfn^2(\mu)$ and the Moyal self-product 
$\Hfn\star\Hfn(\mu)$~\cite{OzorioDeAlmeida1998}.  We have verified this
$\tau^2$ scaling numerically: for Hamiltonians non-diagonal in the
position basis (e.g.\ $\hat{H} = \chi\hbar(\hat{x}+\hat{p})$) the
error ratio $\|\Gfn_{\text{exact}}-\Gfn_{\text{approx}}\|_\infty / \tau^2$
converges to the expected constant 
$\tfrac{1}{2\hbar^2}\|\Hfn\star\Hfn-\Hfn^2\|_\infty$ as $\tau\to 0$.
For Hamiltonians diagonal in a stabilizer basis, 
$\Hfn\star\Hfn = \Hfn^2$ pointwise and the 
approximation~(\ref{eq:shorttime}) is exact at all orders in~$\tau$.
 
Substituting Eq.~(\ref{eq:shorttime}) into the
kernel~(\ref{eq:Kdirect}), the short-time propagator becomes
\begin{multline}
\label{eq:shorttimeK}
\Kprop(\mu_i,\mu_{i-1}) \approx \frac{1}{d^2}
\sum_{\tilde{\mu}\in\Zd\times\Zd}
\omega^{-2\Delta\mu_i\wedge\tilde{\mu}}\times\\
\times
\exp\!\left[\frac{\iu}{\hbar}\big(
\Hfn(\mu_i+\tilde{\mu}) - \Hfn(\mu_{i-1}-\tilde{\mu})
\big)\tau\right],
\end{multline}
where $\Delta\mu_i = \mu_i - \mu_{i-1}$ and
$\mu\wedge\tilde{\mu} = m\tilde{n} - n\tilde{m}$ denotes the
discrete symplectic product for $\mu = (m,n)$ and
$\tilde{\mu} = (\tilde{m},\tilde{n})$.
The direct substitution of Eq.~(\ref{eq:shorttime}) into
Eq.~(\ref{eq:Kdirect}) first produces the phase
$\omega^{+2\Delta\mu_i\wedge\tilde{\mu}}$ with the Hamiltonian difference
$\Hfn(\mu_i-\tilde{\mu})-\Hfn(\mu_{i-1}+\tilde{\mu})$ in the exponent;
the form~(\ref{eq:shorttimeK}) follows by the change of dummy variable
$\tilde{\mu}\to-\tilde{\mu}$, which preserves the full sum over
$\Zd\times\Zd$ and yields Marinov-like Hamiltonian arguments
$\Hfn(\mu_i+\tilde{\mu})-\Hfn(\mu_{i-1}-\tilde{\mu})$ at the cost of
flipping the symplectic phase.  
An equivalent midpoint form is obtained 
by the further shift $\tilde\mu\to\tilde\mu-\Delta\mu_i/2$ (allowed 
since $2^{-1}\in\Zd$ for odd $d$), under which the symplectic phase is 
invariant ($\Delta\mu\wedge\Delta\mu = 0$) and the Hamiltonian arguments 
become $\Hfn(\bar\mu_i\pm\tilde\mu)$ with midpoint $\bar\mu_i = 
(\mu_i+\mu_{i-1})/2$; this is the direct discrete analog of Marinov's 
continuous $H_W(\mathbf{r}+\tilde{\mathbf{r}}/2)-H_W(\mathbf{r}-\tilde{\mathbf{r}}/2)$.
We have verified
identity~(\ref{eq:shorttimeK}) numerically against~(\ref{eq:Kdirect})
(with common $\Gfn\approx e^{-\iu\tau\Hfn/\hbar}$) for
$d=3, 5, 7, 11$.
 
\subsection{Sum over paths}
 
Inserting Eq.~(\ref{eq:shorttimeK}) into the time-sliced
expression~(\ref{eq:slicing}) and collecting the phases, we obtain
\begin{multline}
\label{eq:pathsum}
\Kprop(\mu,t;\mu_0,0) = \left(\frac{1}{d^{2n}}\right)^{\!N}
\sum_{\{\mu_i\}}\sum_{\{\tilde{\mu}_i\}}
\omega^{-2\sum_{i=1}^{N}\Delta\mu_i\wedge\tilde{\mu}_i}\times\\
\times\exp\!\left[\frac{\iu}{\hbar}\sum_{i=1}^{N}
\big(\Hfn(\mu_i + \tilde{\mu}_i)
- \Hfn(\mu_{i-1}-\tilde{\mu}_i)\big)\tau\right],
\end{multline}
where $n$ is the number of qudits (so that the discrete phase space has
$d^{2n}$ points), the first sum runs over all $(N{-}1)$-tuples
$(\mu_1,\ldots,\mu_{N-1}) \in [(\Zd)^{2n}]^{N-1}$ with
$\mu_N = \mu$ fixed, and the second over all $N$-tuples
$(\tilde{\mu}_1,\ldots,\tilde{\mu}_N) \in [(\Zd)^{2n}]^N$.  For
$n=1$ this reduces to $(1/d^2)^N$; for the two-qutrit example of 
Sec.~\ref{sec:beyondDTWA}, $n=2$ and the prefactor becomes 
$(1/d^4)^N = (1/81)^N$.
 
In the limit $N\to\infty$, the sequences $\{\mu_i\}$ and
$\{\tilde{\mu}_i\}$ define piecewise-constant paths on the discrete
phase space.  Introducing the formal notation of paths
$\gamma: [0,t] \to (\Zd)^{2n}$ and
$\xi: [0,t] \to (\Zd)^{2n}$, we write the propagator as
 
\begin{equation}
\label{eq:pathintegral}
\Kprop(\mu,t;\mu_0,0) = \left(\frac{1}{d^{2n}}\right)^{\!N}
\sum_\gamma\sum_\xi
e^{\iu S_{\mathtt{W}}[\gamma,\xi,t]}
\end{equation}
 
\noindent
with the discrete phase-space action
\begin{equation}
\label{eq:discreteaction}
\begin{split}
S_{\mathtt{W}}[\gamma,\xi,t] &=
-\frac{4\pi}{d}\sum_{i=1}^{N}\Delta\gamma_i\wedge\xi_i
\\
&+ \frac{1}{\hbar}\sum_{i=1}^{N}
\big[\Hfn(\gamma_i+\xi_i) - \Hfn(\gamma_{i-1}-\xi_i)\big]\tau\,.
\end{split}
\end{equation}
 
We emphasize that, unlike the continuous Marinov path integral, no 
formal continuum notation for the action (such as 
$\int_0^t\dot{\gamma}\wedge\xi\,\di\tau$) is invoked: the 
action~(\ref{eq:discreteaction}) is the fundamental object, defined as
a finite sum over time steps.  The time variable is continuous (the limit
$N\to\infty$ is taken), but the dynamical variables $\gamma_i$ and $\xi_i$
are elements of $(\Zd)^{2n}$ at each time step.
 
Equations~(\ref{eq:pathintegral})--(\ref{eq:discreteaction})
constitute the main result of this paper.  They provide the exact discrete
analog of Marinov's path integral
(\ref{eq:Marinov})--(\ref{eq:Marinovaction}) for the Wigner function
propagator.  Several features deserve comment:
 
\begin{enumerate}
\item The ``kinetic'' term $\Delta\gamma_i\wedge\xi_i$ involves the discrete
symplectic form, replacing $\dot{\mathbf{r}}\wedge\tilde{\mathbf{r}}\,\di t$ in the
continuous case.
\item The ``potential'' term
$\Hfn(\gamma_i+\xi_i) - \Hfn(\gamma_{i-1}-\xi_i)$ is the discrete
finite-difference analog of the continuous expression
$H_W(\mathbf{r}+\frac{1}{2}\tilde{\mathbf{r}})
-H_W(\mathbf{r}-\frac{1}{2}\tilde{\mathbf{r}})$.
\item The sum runs over piecewise-constant paths on the discrete toroidal
lattice $(\Zd)^{2n}$.  Unlike the continuous case, where the path integral
notation $\int\mathcal{D}r$ denotes a limit of finite-dimensional integrals,
here both the phase space \emph{and} the set of paths are inherently discrete:
for each time-slice, each path variable takes one of the $d^{2n}$ phase-space values.
At any finite $N$, expression~(\ref{eq:pathsum}) is a well-defined finite
sum; the exact propagator is obtained in the limit $N\to\infty$, which
refines only the time discretization and not the phase-space lattice.
For Hamiltonians diagonal in a stabilizer basis (such as those of 
Sec.~\ref{sec:qutrit}), the symbol-level short-time
relation~(\ref{eq:shorttime}) is exact, and~(\ref{eq:pathsum}) then
coincides with the exact propagator at any~$N$; we have verified this
numerically for $d=3,5,7,11$ by comparing single-step ($N=1$) and 
iterated ($N=2,3,\ldots$) evaluations of~(\ref{eq:pathsum}) against 
direct $\hat U\hat\rho\hat U^\dagger$ evolution.
\item The normalization factor $(1/d^{2n})^N$ plays the role of the
measure in the continuous path integral, scaling as the inverse of
the discrete phase-space volume raised to the number of time slices.
\item In the $\xi_i = 0$ sector, every symplectic phase is unity
and the Hamiltonian piece of the action telescopes to the boundary 
term $\tau[\Hfn(\gamma_N)-\Hfn(\gamma_0)]/\hbar$.  After the free sum 
over the $N-1$ intermediate $\gamma_i$ (each contributing a factor 
$d^{2n}$), the $\xi=0$ contribution admits the closed form
\begin{equation}
\begin{split}
\label{eq:xizero}
\Kprop^{\xi=0}(\gamma_N,t;\gamma_0,0) &= 
 \\ 
\frac{1}{d^{2n} }
\exp & \!\left[\frac{\iu t}{N\hbar}\big(\Hfn(\gamma_N)-\Hfn(\gamma_0)\big)\right],
\end{split}
\end{equation}

with $\tau = t/N$.  This kernel is non-real at finite $N$ whenever 
$\Hfn(\gamma_N)\ne\Hfn(\gamma_0)$, and, since the exponent scales as 
$1/N$, collapses in the continuum limit $N\to\infty$ to a trivial 
uniform kernel $1/d^{2n}$ independent of both endpoints.  In both 
regimes---finite $\tau$ (non-real) and continuum ($\tau\to 0$, 
uniform and trivial)---the $\xi=0$ sector alone fails to reproduce 
the propagator; the full quantum dynamics requires the coherent 
contribution of the $\xi \neq 0$ sectors (see 
Sec.~\ref{sec:beyondDTWA}).
\end{enumerate}
 
\subsection{Linearized action and pseudo-classical dynamics}
\label{sec:classicallimit}
 
For Hamiltonians whose phase-space function is affine in the 
coordinates, $\Hfn(\mu) = a\,m + b\,n + c$, with 
$\mu=(m,n)\in\Zd\times\Zd$ and integer-valued representatives 
$m, n\in\{0,1,\ldots,d-1\}$, we treat $\Hfn$ as the restriction to 
this fundamental domain of a linear function on the lift 
$\mathbb{Z}\times\mathbb{Z}$.  With this convention, the symbol-level 
identity
\begin{equation}
\label{eq:linearH_diff}
\Hfn(\gamma_i+\xi_i) - \Hfn(\gamma_{i-1}-\xi_i) =
(a,b)\cdot(\Delta\gamma_i + 2\xi_i)
\end{equation}
holds as a relation between integer-valued symbols, provided one uses 
the lifted arithmetic on the right-hand side.  Substituting into 
Eq.~(\ref{eq:shorttimeK}), the $\tilde\mu$-sum becomes
\begin{equation}
\sum_{\tilde n_i\in\Zd}
\omega^{-(2\Delta m_i)\tilde n_i}\,
e^{2\iu b\tau \tilde n_i/\hbar}\,
\sum_{\tilde m_i\in\Zd}
\omega^{(2\Delta n_i)\tilde m_i}\,
e^{2\iu a\tau \tilde m_i/\hbar}.
\end{equation}
For the two factors above to collapse to Kronecker deltas, both the 
Gauss-sum arguments \emph{and} the phases generated by summing a 
lifted integer symbol over a period must be periodic modulo $d$.  The 
first condition requires $2a\tau/\hbar,\;2b\tau/\hbar \in 
(2\pi/d)\mathbb{Z}$.  The second requires that the residual phase 
$e^{-\iu a\tau d/\hbar}=(-1)^{k_a}$ (and its $b$-counterpart) that 
arises from integer wraparound of the lifted symbol across its period 
evaluates to $+1$, i.e.\ $k_a,k_b$ \emph{even}, where $k_a = 
ad\tau/(\pi\hbar)$ and $k_b = bd\tau/(\pi\hbar)$.  The two conditions 
combine into the single ``strict commensurability'' requirement
\begin{equation}
\label{eq:strictcommen}
\frac{a\tau}{\hbar},\;\frac{b\tau}{\hbar}\;\in\;\frac{2\pi}{d}\mathbb{Z}.
\end{equation}
When Eq.~(\ref{eq:strictcommen}) holds, the $\tilde\mu$-sum collapses 
to Kronecker deltas enforcing the discrete analog of Hamilton's 
equations,
\begin{equation}
\label{eq:discreteHamilton}
\Delta m_i \equiv k_b/2 \pmod d,\qquad
\Delta n_i \equiv -k_a/2 \pmod d,
\end{equation}
with $k_a,k_b\in 2\mathbb{Z}$ so that $k_a/2, k_b/2$ are ordinary 
integers.  The propagator then reduces to a deterministic shift on 
the toroidal lattice.  Equivalently, Eq.~(\ref{eq:strictcommen}) is 
the condition that $\hat U = e^{-\iu\tau\hat{H}/\hbar}$ acts as a pure 
Pauli (Heisenberg) operator on the qudit, which by Gross's 
theorem~\cite{Gross2006} induces an affine symplectic shift on the 
discrete Wigner function.
When only the weaker condition $2a\tau/\hbar,\;2b\tau/\hbar\in 
(2\pi/d)\mathbb{Z}$ holds (i.e.\ $k_a$ or $k_b$ integer but odd), the 
wraparound phase $(-1)^{k_a}$ does not cancel, $\hat U$ is a 
``half-Pauli'' root outside the Clifford group, and the propagator 
develops a non-delta structure---the kernel is no longer supported 
on a single lattice shift but spreads over $\Zd\times\Zd$ with an 
oscillating signature.  Away from even commensurability, the 
$\tilde\mu$-sum gives an exponential-sum pattern on $\Zd$ that 
broadens the kernel across the lattice.  For \emph{nonlinear} 
Hamiltonians, the cancellation in~(\ref{eq:linearH_diff}) fails and 
the coupling between $\gamma$ and $\xi$ sectors generates the 
genuinely quantum interference exemplified in 
Sec.~\ref{sec:beyondDTWA}.  The generalization to $n$ qudits with 
$\Hfn(\mu) = \sum_{a=1}^{n}(a_a m_a + b_a n_a) + c$ is immediate: the 
$\tilde\mu$-sum factorizes into $2n$ independent sums and each 
subsystem $a$ carries its own strict commensurability condition 
$a_a\tau/\hbar,\;b_a\tau/\hbar\in(2\pi/d)\mathbb{Z}$.  When all 
subsystems are simultaneously strictly commensurate, the composite 
kernel reduces to a product of Pauli shifts; a single 
non-commensurate or non-linear degree of freedom is already 
sufficient to reintroduce the full fluctuation sum.

\section{Applications: Qutrit Systems}
\label{sec:qutrit}
 
We now illustrate the formalism with the smallest odd-prime system,
the qutrit ($d = 3$), for which $\omega = e^{2\pi\iu/3}$ and
$2^{-1} = 2 \pmod{3}$.
 
\subsection{Single qutrit under a diagonal Hamiltonian}
\label{sec:singlequtrit}
 
Consider a three-level system with Hamiltonian
\begin{equation}
\label{eq:qutritH}
\hat{H} = \hbar\chi\,\hat{x} = \hbar\chi\,\mathrm{diag}(0,1,2)\,,
\end{equation}
which models, e.g., equally-spaced energy levels with splitting $\hbar\chi$.
The evolution operator is $\hat{U}(t) = \mathrm{diag}(1,e^{-\iu\chi t},e^{-2\iu\chi t})$.
 
We take the initial state to be the uniform superposition
$\ket{\psi(0)} = \frac{1}{\sqrt{3}}(1,1,1)^T$, which is the
momentum eigenstate $\ket{p,0}$.  Its density matrix is
$\hat{\rho}(0) = \frac{1}{3}J$, where $J$ is the $3\times 3$ matrix
of all ones.  A direct computation of the Weyl symbols
$\tilde{\rho}(k,j) = \Tr[\hat{\rho}\,\hat{D}(k,j)]$ using
$\hat{D}(k,j) = \omega^{-kj\cdot 2^{-1}}\hat{Z}^k\hat{X}^j$
yields $\tilde{\rho}(k,j) = \delta_{k,0}$ (independent of~$j$).
Substituting into Eq.~(\ref{eq:DWF}) gives
\begin{equation}
\Wfn_0(m,n) = \frac{1}{3}\delta_{n,0}\,,
\end{equation}
i.e., the Wigner function is $1/3$ on the three points $(m,0)$ with
$m\in\{0,1,2\}$, and zero elsewhere on the $3\times 3$ lattice.
 
To compute the time-evolved Wigner function, we need the phase-space
function $\Gfn(m,n) = \mathscr{F}^{-1}[\tilde{\Gfn}](m,n)$
(with $1/d$ normalization as in Eq.~(\ref{eq:phasespacefunction})),
where $\tilde{\Gfn}(k,j) = \Tr[\hat{U}(t)\hat{D}(k,j)]$.
 
Substituting into the propagator~(\ref{eq:Kdirect}) and summing over the
initial Wigner function, one obtains the time-evolved Wigner function.
At $\chi t = \pi$, for example, explicit computation gives
\begin{equation}
\Wfn(m,n,\pi/\chi) = \begin{cases}
-\tfrac{1}{9} & (m,n)\in\{(0,0),(2,0)\}\\
\tfrac{1}{3} & (m,n) = (1,0)\\
\tfrac{2}{9} & m\in\{0,2\},\ n\in\{1,2\}\\
0 & \text{otherwise}\,,
\end{cases}
\end{equation}
which sums to unity and exhibits negativity at $(0,0)$ and $(2,0)$
(Fig.~\ref{fig:wigner}).

\begin{figure*}[t]
\centering
\includegraphics[width=0.82\textwidth]{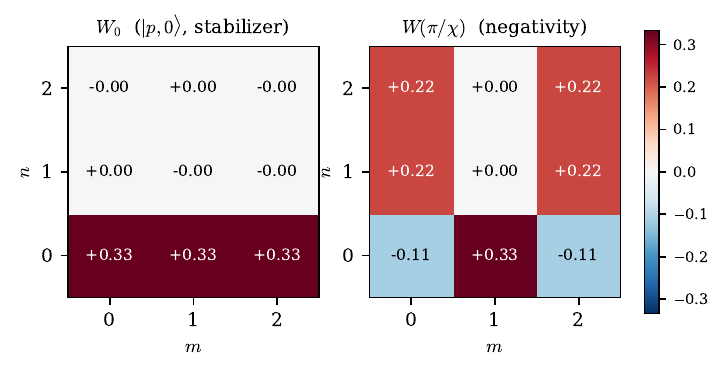}
\caption{\label{fig:wigner}%
Discrete Wigner function of the single qutrit ($d=3$) under
$\hat H = \hbar\chi\hat x$.  Left: the initial momentum eigenstate
$\ket{p,0}$, a stabilizer state, has a non-negative Wigner function
supported on the line $n=0$.  Right: at $\chi t = \pi$ the evolved
state $\hat U(\pi/\chi)\ket{p,0}$ is no longer a stabilizer state
[$\hat U(\pi/\chi)=\mathrm{diag}(1,-1,1)$ lies outside the qutrit
Clifford group] and, in accordance with the converse of Gross's
theorem~\cite{Gross2006}, develops negative Wigner values $-\tfrac{1}{9}$
at $(0,0)$ and $(2,0)$.  Values follow from the propagator
Eq.~(\ref{eq:Kdirect}) and are reproduced by the path
integral~(\ref{eq:pathsum}) to machine precision.}
\end{figure*}
Since $\hat H$ is diagonal in the stabilizer (position) basis, the
symbol-level short-time relation~(\ref{eq:shorttime}) is exact and 
Eq.~(\ref{eq:pathsum}) at any finite $N$ reproduces this Wigner 
function; we have confirmed agreement with both the direct evolution 
$\hat U\hat\rho_0\hat U^\dagger$ and the path integral for $N=1,2,3,4$ 
to machine precision.
The initial state $\ket{p,0}$ is a stabilizer state of the qutrit
and hence has a non-negative Wigner function.  The evolution operator
$\hat U(\pi/\chi) = \mathrm{diag}(1,-1,1)$ is not in the qutrit 
Clifford group: its diagonal entries are square roots of unity rather 
than cube roots, so it does not normalize the Pauli 
group~\cite{Gottesman1998}.  Consequently the evolved state 
$\hat U(\pi/\chi)\ket{p,0}$ is not a stabilizer state, and by the 
converse of Gross's theorem~\cite{Gross2006} its discrete Wigner 
function is allowed to take negative values, as observed.
 
\subsection{Fluctuation sectors and entanglement: two interacting qutrits}
\label{sec:beyondDTWA}

To exhibit the role of the nonzero fluctuation sectors in an interacting
model that generates entanglement, we consider two interacting qutrits
with Hamiltonian
\begin{equation}
\label{eq:twoqutrit}
\hat{H} = \hbar\chi\left(\hat{x}_1\otimes\hat{x}_2\right),
\end{equation}
which generates entanglement between the two subsystems.  The Hilbert
space is $\mathscr{H} = \mathbb{C}^3\otimes\mathbb{C}^3$, and the
discrete phase space is $(\mathbb{Z}_3)^4$ with $3^4 = 81$ points.
 
The path integral~(\ref{eq:pathsum}) for the composite system reads
\begin{multline}
\label{eq:pathsum2}
\Kprop(\mu,t;\mu_0,0) = \left(\frac{1}{3^4}\right)^{\!N}
\sum_{\{\mu_i\}}\sum_{\{\tilde{\mu}_i\}}
\omega^{-2\sum_i\Delta\mu_i\wedge\tilde{\mu}_i}\times\\
\times\exp\!\left[\frac{\iu}{\hbar}\sum_i
\big(\Hfn(\mu_i + \tilde{\mu}_i)
- \Hfn(\mu_{i-1}-\tilde{\mu}_i)\big)\tau\right],
\end{multline}
where now $\mu_i = (m_1,n_1,m_2,n_2)_i$ and the symplectic product
is $\mu\wedge\tilde{\mu} = \sum_{a=1}^{2}(m_a\tilde{n}_a - n_a\tilde{m}_a)$.
 
The $\tilde{\mu}=0$ sector of the path integral has a simple
structure: each symplectic phase $\omega^{-2\Delta\mu_i\wedge 0}=1$,
and the Hamiltonian sum telescopes to a boundary term in the action.  
Summing freely over the intermediate $\mu_i$ with $0<i<N$ collapses 
the $\tilde\mu=0$ contribution to the explicit closed form of 
Eq.~(\ref{eq:xizero}), $\Kprop^{\xi=0}(\mu_N;\mu_0) = d^{-2n}
\exp[\iu t\{\Hfn(\mu_N)-\Hfn(\mu_0)\}/(N\hbar)]$.
On its own, this sector fails to furnish the correct Wigner 
propagator in two distinct ways: (i) at any finite time step $\tau$, 
it has a non-vanishing imaginary part $\sim\tau$ whenever 
$\Hfn(\mu_N)\ne\Hfn(\mu_0)$, while $\Kprop$ is rigorously real 
(cf.\ Sec.~\ref{sec:kernelderivation}); (ii) in the continuum limit 
$\tau\to 0$, the boundary phase trivializes to unity, the $\xi=0$ 
contribution reduces to the uniform kernel $d^{-2n}$ (independent of 
$\mu_0$ and $\mu$) and therefore carries no dynamical information 
whatsoever.  Only the coherent addition of the $\tilde\mu\neq 0$ 
sectors---whose mutual interference is the discrete counterpart of 
the phase-space interference pattern that in the continuous case 
yields the Moyal bracket beyond the Poisson bracket---restores both 
reality and the non-trivial quantum dynamics.  We have checked both 
failure modes numerically for the two-qutrit system: retaining only 
$\tilde\mu=0$ at $N=1$ yields 
$|\mathrm{Im}\,\Kprop|_{\max}\sim 10^{-2}$ at $\chi\tau=0.5$, while 
for $N\geq 8$ the imaginary part is suppressed as $\sim 1/N$ but the 
real part tends to a uniform value $1/81\approx 0.0123$ for all 
$(\mu',\mu)$, reproducing none of the dynamics of the full 
propagator.  The discrete truncated Wigner approximation 
(DTWA)~\cite{Schachenmayer2015} is a separate semiclassical scheme 
that samples the Wigner function and propagates each sample along 
mean-field trajectories, capturing two-point correlations at short 
times but not, in general, higher-order observables such as 
entanglement entropy.
 
At each time step, the exact kernel sums over all
$d^{2n}$ fluctuation vectors in the composite phase space
($n$ the number of qudits, $d^{2n}=81$ for the present two-qutrit system):
\begin{multline}
\label{eq:DTWAcorrection}
\Kprop(\mu',\mu;\tau) = \frac{1}{81}
\sum_{\tilde{\mu}\in(\mathbb{Z}_3)^4}
\omega^{-2\Delta\mu\wedge\tilde{\mu}}\times\\
\times\exp\!\left[\frac{\iu}{\hbar}
\big(\Hfn(\mu'+\tilde{\mu}) - \Hfn(\mu-\tilde{\mu})\big)\tau\right],
\end{multline}
of which $80$ are genuine fluctuation contributions beyond the
boundary-term $\tilde{\mu}=0$ sector.
 
As a concrete quantitative test, we consider the initial product state
$\ket{\psi(0)} = \ket{p,0}\otimes\ket{p,0}$, for which the
Wigner function is the product
$\Wfn_0(\mu) = \Wfn^{(1)}_0(m_1,n_1)\Wfn^{(2)}_0(m_2,n_2)$,
and compute the linear entropy
$S_L(t) = 1 - \Tr[\hat{\rho}_1^2(t)]$ as a witness of entanglement
generation.  The time-evolved state is
$\ket{\psi(t)} = \frac{1}{3}\sum_{m,n}e^{-\iu\chi t\,mn}\ket{x,m}\ket{x,n}$.
The reduced density matrix has elements
$(\hat{\rho}_1)_{mm'} = \frac{1}{9}\sum_{n=0}^{2}e^{-\iu\chi t\,n(m-m')}$,
from which
\begin{equation}
\label{eq:traceRhoSq}
\Tr[\hat{\rho}_1^2] = \frac{27 + 4(1+2\cos\chi t)^2 + 2(1+2\cos 2\chi t)^2}{81}.
\end{equation}
Expanding for short times yields
\begin{equation}
\label{eq:linearentropy}
S_L(t) = \frac{8\chi^2 t^2}{9} + O(t^4)\,,
\end{equation}
which we have independently verified against exact diagonalization of
the $9\times 9$ Hamiltonian matrix.
We have also evaluated Eq.~(\ref{eq:pathsum2}) numerically for this
system and confirmed that the path integral reproduces the full
closed-form expression~(\ref{eq:traceRhoSq}) for $\Tr[\hat\rho_1^2]$ 
to numerical precision (we find $S_L = 8.823\times 10^{-3}$ at 
$\chi t = 0.1$, agreeing with exact diagonalization to ten decimals).
Any truncation restricted to a single $\tilde\mu$ sector---most 
natively $\tilde\mu=0$, discussed above---fails to reproduce the 
$\Tr[\hat\rho_1^2]$ dynamics: the quantum entanglement generation 
encoded in~(\ref{eq:traceRhoSq}) arises from the coherent sum over 
all fluctuation sectors.
The coefficient $8/9 = 2[\mathrm{Var}(\hat{x})]^2$ with
$\mathrm{Var}(\hat{x}) = 2/3$ for the state $\ket{p,0}$ is
consistent with the general short-time perturbation-theory result
$S_L(t) \approx 2\chi^2 t^2\,
\mathrm{Var}(\hat{A})\mathrm{Var}(\hat{B})$
for a bipartite interaction $\hat{H} = \chi\hat{A}\otimes\hat{B}$
acting on a product state.
The DTWA in its standard implementation~\cite{Schachenmayer2015}
is known to reproduce one- and two-point correlations at short times
but to be inexact for higher-order observables such as the linear
entropy computed here.  The $\tilde{\mu}\neq 0$ fluctuation sectors in
Eq.~(\ref{eq:DTWAcorrection}) are precisely the contributions that
the path integral formulation makes explicit and that any
mean-field-type truncation necessarily misses.
 
We note that Eq.~(\ref{eq:linearentropy}) is merely the short-time
expansion of the exact closed-form expression~(\ref{eq:traceRhoSq})
for $\Tr[\hat\rho_1^2]$, which is valid for all times.
The linear entropy is periodic with $S_L(t + 2\pi/\chi) = S_L(t)$,
symmetric about $\chi t = \pi$, and reaches its global maximum 
$S_L^{\max} = 2/3$ (the maximum entropy of a maximally mixed qutrit) 
at the two values $\chi t = 2\pi/3$ and $\chi t = 4\pi/3$ within a 
period, where both $(1+2\cos\chi t)$ and $(1+2\cos 2\chi t)$ vanish 
simultaneously.  The entropy vanishes at the revival $\chi t = 2\pi$
where the system returns to the initial product state.  The full time
dependence is shown in Fig.~\ref{fig:entropy}; representative
values,
\begin{equation}
\label{eq:SLvalues}
\begin{array}{c|ccccccc}
\chi t & 0.25 & 0.5 & \pi/2 & 2\pi/3 & \pi & 4\pi/3 & 2\pi \\\hline
S_L(t) & 0.053 & 0.185 & 0.593 & 0.667 & 0.395 & 0.667 & 0.000
\end{array}
\end{equation}
are reproduced by the path integral~(\ref{eq:pathsum2}) to machine
precision; the short-time expansion $8\chi^2 t^2/9$ is accurate only
for $\chi t \lesssim 0.3$, while the full path integral captures the
entire non-trivial time dependence.

\begin{figure}[t]
\centering
\includegraphics[width=\columnwidth]{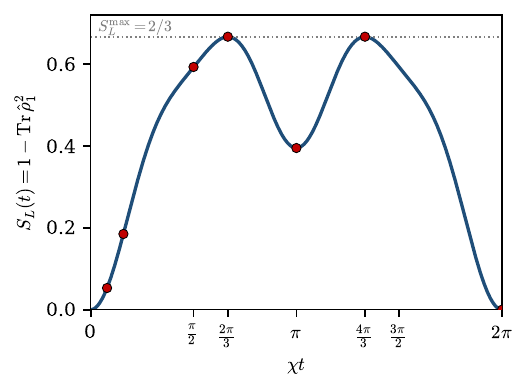}
\caption{\label{fig:entropy}%
Linear entropy $S_L(t) = 1 - \Tr[\hat\rho_1^2]$ of one qutrit in the
interacting two-qutrit system $\hat H = \hbar\chi(\hat x_1\otimes\hat
x_2)$, starting from the product state $\ket{p,0}\otimes\ket{p,0}$
(solid curve, exact closed form Eq.~(\ref{eq:traceRhoSq}); markers, the
tabulated values, each reproduced by the path
integral~(\ref{eq:pathsum2}) to machine precision).  The entropy is
periodic and symmetric about $\chi t = \pi$, reaching the maximal-mixing
value $S_L^{\max}=2/3$ (dotted line) at $\chi t = 2\pi/3$ and $4\pi/3$,
and vanishing at the revival $\chi t = 2\pi$.  Capturing the dynamics
beyond the initial $O(t^2)$ rise requires the coherent sum over all
fluctuation sectors $\tilde\mu\neq 0$.}
\end{figure}

\subsection{A genuinely approximate case: convergence of the path integral}
\label{sec:convergence}

In both examples above the Hamiltonian is diagonal in a stabilizer
(position) basis, so that $\Hfn\star\Hfn = \Hfn^2$ pointwise, the
symbol-level short-time relation~(\ref{eq:shorttime}) is \emph{exact},
and the path integral~(\ref{eq:pathsum}) already reproduces the exact
dynamics at any finite number of slices $N$.  To exhibit the path
integral in the regime for which it is designed---where the short-time
kernel is a genuine approximation that becomes exact only as
$N\to\infty$---we consider a single qutrit with the nonlinear,
non-diagonal Hamiltonian
\begin{equation}
\label{eq:harperH}
\hat H = \hbar\chi\left(\hat x^2 + \hat p^2\right),
\end{equation}
the finite-dimensional analog of a harmonic oscillator.  Since
$\hat x^2$ is diagonal in the position basis and $\hat p^2$ in the
momentum basis, $\hat H$ is not diagonal in any single stabilizer
basis; consequently $\Hfn\star\Hfn \neq \Hfn^2$ and the short-time
approximation~(\ref{eq:shorttime}) carries the generic $O(\tau^2)$
error of Eq.~(\ref{eq:shorttimeerror}).  We have verified this scaling
directly: the symbol error
$\|\Gfn_{\text{exact}}-e^{-\iu\tau\Hfn/\hbar}\|_\infty/\tau^2$
converges to the constant
$\tfrac{1}{2}\|\Hfn\star\Hfn-\Hfn^2\|_\infty \approx 8.16$
(in units $\hbar=\chi=1$) as $\tau\to 0$.

Taking the stabilizer initial state $\ket{x,0}$ and evolving to
$\chi t = 0.6$, the exact discrete Wigner function develops pronounced
negativity (minimum value $\approx -0.21$; Fig.~\ref{fig:convergence},
left), confirming that the evolution leaves the stabilizer polytope and
that the dynamics is genuinely non-classical.  Iterating the short-time
propagator~(\ref{eq:shorttimeK}) over $N$ slices, the resulting
Wigner function $W^{(N)}(\chi t)$ converges monotonically to the exact
result, with
$\|W^{(N)}-W_{\text{exact}}\|_\infty$ falling from $0.43$ at $N=1$ to
$1.3\times 10^{-2}$ at $N=64$---a clean $O(1/N)$ decay (log--log slope
$-0.98$; Fig.~\ref{fig:convergence}, right), exactly as expected from
an $O(\tau^2)$ local error accumulated over $N = t/\tau$ steps.  This
example demonstrates that the discrete path integral reproduces the
full quantum evolution, including the dynamical generation of Wigner
negativity, in the non-trivial regime where no finite-$N$ truncation is
exact.

\begin{figure*}[t]
\centering
\includegraphics[width=0.85\textwidth]{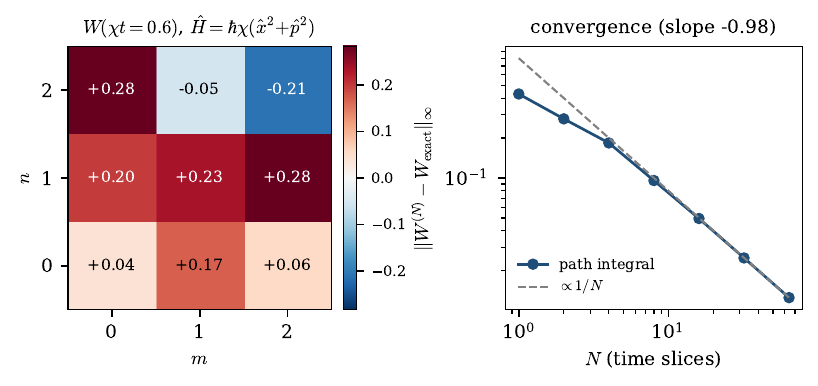}
\caption{\label{fig:convergence}%
Non-stabilizer single qutrit with $\hat H = \hbar\chi(\hat x^2+\hat
p^2)$, evolved from $\ket{x,0}$.  Left: the exact discrete Wigner
function at $\chi t = 0.6$ exhibits negativity (minimum $\approx
-0.21$), so the state lies outside the stabilizer set and the dynamics
is non-classical.  Right: convergence of the iterated short-time path
integral~(\ref{eq:shorttimeK}) to the exact Wigner function as the
number of time slices $N$ increases.  Because $\hat H$ is not
stabilizer-diagonal, the short-time symbol~(\ref{eq:shorttime}) is only
$O(\tau^2)$ accurate and the path integral is exact only in the limit
$N\to\infty$; the error decays as $O(1/N)$ (dashed guide), in contrast
to the stabilizer-diagonal examples of
Secs.~\ref{sec:singlequtrit}--\ref{sec:beyondDTWA}, which are exact at
any~$N$.}
\end{figure*}

\section{Discussion}
\label{sec:applications}
 
\subsection{Relation to the DTWA and prospects for corrections}

The two-qutrit example of Sec.~\ref{sec:beyondDTWA} illustrates a
general principle: the path integral~(\ref{eq:pathsum}) provides the
exact discrete-phase-space propagator, of which the DTWA and related
semiclassical methods~\cite{Schachenmayer2015} are inequivalent
approximations motivated by mean-field considerations.  The
contribution of the $\xi\neq 0$ fluctuation sectors is precisely
what generates the non-trivial entanglement dynamics.  We stress that
the present work makes these sectors \emph{explicit} but does not yet
furnish a tractable corrected scheme; the path integral is here a
formal and conceptual organizing principle, and a practical
finite-density method remains contingent on the sign-problem
considerations discussed below.  For
many-body systems of $n$ qudits, the phase space has $d^{2n}$ points
at each time step, and an exhaustive enumeration of all fluctuation
paths over $N$ time slices scales as $(d^{2n})^N$---clearly
intractable for large $n$.  Monte Carlo sampling of the fluctuation
paths $\xi$, analogous to the way Marinov's continuous path integral 
underlies Monte Carlo schemes for Wigner 
dynamics~\cite{Dittrich2010,Marinov1991}, offers a possible numerical 
strategy; however, the oscillatory complex weights 
$e^{\iu\mathcal{S}}$ give rise to a sign problem that must be 
addressed, for instance by saddle-point expansion around stationary 
$(\gamma,\xi)$ configurations or by resummation techniques, before 
such sampling becomes efficient.
 
\subsection{Characterization of non-classicality}
 
Gross's discrete Hudson theorem~\cite{Gross2006} establishes that, for
systems of odd prime dimension, the only pure states with a non-negative
discrete Wigner function are the stabilizer states.  This result, combined
with the Gottesman--Knill theorem~\cite{Gottesman1998}, implies that
Wigner negativity is a necessary resource for quantum computational
advantage~\cite{Veitch2012,Howard2014,Mari2012}, a connection 
made quantitative in semiclassical simulation schemes based on 
Gauss-sum evaluation of discrete Wigner path 
integrals~\cite{KociaLove2021}.
 
The propagator $\Kprop$ provides a dynamical perspective on this
resource theory.  Since $\Kprop$ itself can take negative values,
the generation or destruction of Wigner negativity under time evolution
can be tracked through the kernel.  The path integral
representation~(\ref{eq:pathintegral}) exhibits this structure
explicitly as a sum of complex phases over fluctuation sectors $\xi$,
suggesting that interference between these sectors plays a role in
the dynamical growth of Wigner negativity and ``magic''
(non-stabilizerness).  Making this connection quantitative is a
natural direction for further work.
 
\subsection{Composite systems and entanglement}
 
For composite systems of $n$ qudits, the phase space is
$(\Zd)^{2n}$ and the formalism generalizes straightforwardly.  The
path integral involves $2n$-component paths, and the action contains
the full many-body Hamiltonian.  Entanglement dynamics in this setting
has a richer structure than one might initially expect: by Gross's 
theorem~\cite{Gross2006}, entangled \emph{stabilizer} states 
(e.g.\ the qutrit Bell-type states $|\Phi\rangle = 
d^{-1/2}\sum_{n}|n,n\rangle$) have \emph{non-negative} discrete 
Wigner functions despite being maximally entangled, while entangled 
non-stabilizer states necessarily exhibit negativity (for two-qubit 
systems, a direct connection between entanglement and Wigner 
negativity was established in~\cite{Franco2006}).
The path integral provides a tool for studying the dynamical 
generation and propagation of both types of entanglement, and the 
associated (de)localization of Wigner negativity, through the 
interference of phase-space paths, as we demonstrated concretely for 
the two-qutrit system.
 
\section{Conclusions and Outlook}
\label{sec:conclusions}
 
We have constructed a path integral formulation for the dynamics of
quantum systems with a finite-dimensional Hilbert space of odd prime
dimension $d$, working entirely within the discrete phase space
$\Zd \times \Zd$.  The central result, Eq.~(\ref{eq:pathintegral}),
expresses the propagator of the discrete Wigner function as a sum
over piecewise-constant paths on the toroidal lattice, weighted by a
discrete phase-space action functional~(\ref{eq:discreteaction}).
 
The key steps in the construction were: (i) building the discrete Weyl
transform from generalized displacement operators derived from the
discrete Fourier transform; (ii) deriving the evolution kernel using
the twisted convolution (discrete Moyal product); and (iii) iterating
the composition law of the kernel with a short-time approximation to
obtain the path sum.
 
The formulation was verified for a single qutrit, and its utility
illustrated for an interacting two-qutrit system.  Beyond the 
short-time $O(t^2)$ growth of the linear entropy, the exact 
closed-form expression for $\Tr[\hat\rho_1^2]$ is reproduced by
the full path integral at all times, including the twin global maxima
$S_L^{\max} = 2/3$ at $\chi t = 2\pi/3$ and $4\pi/3$ and the revival 
at $\chi t = 2\pi$.  The $\tilde\mu = 0$ truncation of the 
fluctuation sum fails in two complementary ways---it violates 
reality at finite time step and reduces to a trivial uniform kernel 
in the continuum limit---showing that genuine quantum entanglement
requires the coherent sum over all fluctuation sectors.
For a non-stabilizer Hamiltonian, $\hat H = \hbar\chi(\hat x^2+\hat
p^2)$, for which the short-time kernel is only $O(\tau^2)$ accurate, we
demonstrated that the iterated path integral converges to the exact
dynamics---including the dynamically generated Wigner negativity---as
$O(1/N)$ in the number of time slices $N$, confirming the formulation
in the non-trivial regime where no finite-$N$ truncation is exact.
We also identified an exactly solvable ``pseudo-classical'' regime:
for Hamiltonians linear in the phase-space coordinates and at 
\emph{strictly} commensurate times---those for which the single-step 
evolution~$\hat U$ lies in the Pauli (Heisenberg) group---the 
fluctuation sum collapses to a deterministic shift on the toroidal 
lattice, realizing the discrete analog of classical Hamiltonian 
flow.  Away from strict commensurability, or for any non-linear 
contribution to the Hamiltonian, the full fluctuation sum is 
required and the propagator carries the non-classical interference 
responsible for Wigner negativity.
 
Several directions for future work suggest themselves.  First,
the extension to $d = 2$ (qubits) requires replacing the displacement
operators by a construction that avoids the multiplicative inverse 
of~$2$; the framework of~\cite{Raussendorf2017,Delfosse2015} provides 
a starting point.  Second, extending the formalism to prime-power
dimensions $d = p^k$ via the Galois-field
construction~\cite{Gibbons2004} would broaden its applicability.
Third, Monte Carlo sampling of the fluctuation paths $\xi$ in the
path sum offers a route toward numerical schemes for quantum
corrections to the DTWA in many-body systems of interacting
qudits~\cite{Zhu2019}, provided the sign problem inherent to the
oscillatory phase $e^{\iu\mathcal{S}}$ can be tamed by saddle-point 
or resummation techniques.  Finally, extending the formalism to open
quantum systems~\cite{PachonIngoldDittrich2010} would provide a 
phase-space path integral description of decoherence in 
finite-dimensional systems.

\section*{Code Availability}
The numerical routines that reproduce every quantitative claim of this 
paper---displacement-operator and Wigner-function construction, the 
three equivalent forms of the propagator (trace form, 
Eq.~(\ref{eq:Kdirect}), and the iterated short-time path integral), the 
single-qutrit Wigner function at $\chi t=\pi$, the closed-form purity
of Eq.~(\ref{eq:traceRhoSq}), the linear-entropy table, the
strict-commensurability analysis of
Sec.~\ref{sec:classicallimit}, and the $O(1/N)$ convergence of the path
integral for the non-stabilizer Hamiltonian of
Sec.~\ref{sec:convergence}---are collected in an open Python
repository~\cite{quditMarinov2026}, available at
\url{https://github.com/lpachon/qudit-marinov}.  Each claim is 
implemented as a standalone script that exits successfully only if the 
associated numerical test matches to machine precision (typically 
$10^{-10}$) for $d\in\{3,5,7,11\}$ where applicable.  The code depends 
only on NumPy and SciPy and can be run with a single command 
(\verb|python run_all_tests.py|). 

\section*{Acknowledgments}
This work was supported by the R+D+I efforts from guane Enterprises.
\bibliographystyle{apsrev4-2}
\bibliography{fdHspd}

\end{document}